\documentclass[10pt,twocolumn]{article}

\usepackage[utf8]{inputenc}
\usepackage{amsmath}
\usepackage{amsthm}
\usepackage{amssymb}
\usepackage{titlesec}
\usepackage{cite}
\usepackage{booktabs}
\usepackage{graphicx}
\usepackage[colorlinks=false]{hyperref}
\usepackage[format=plain,labelfont=it]{caption}
\usepackage[left=1.5cm,right=1.5cm,top=2cm,bottom=2cm]{geometry}
\usepackage{color}
\usepackage{fancyhdr}

\titleformat{\section}{\centering\normalfont\scshape}{\Roman{section}.}{5pt}{}
\titleformat{\subsection}{\normalfont\it}{\Alph{subsection}.}{5pt}{}
\titleformat{\subsubsection}{\normalfont\it}{\hspace{4mm}\arabic{subsubsection})}{5pt}{}

\newcommand\infoFootnote[1]{%
  \begingroup
  \renewcommand\thefootnote{}\footnote{#1}%
  \addtocounter{footnote}{-1}%
  \endgroup}

\newtheorem{thm}{Theorem}

\newtheorem{lem}[thm]{Lemma}

\newtheorem{assum}{Assumption}

\newtheorem{defn}{Definition}

\newtheorem{rem}{Remark}

\DeclareMathOperator*{\argmin}{\arg\min}

\newcommand{\R}{\mathbb{R}} 

\newcommand{\ab}{\boldsymbol{a}}
\newcommand{\fb}{\boldsymbol{f}}
\newcommand{\gb}{\boldsymbol{g}}
\newcommand{\ub}{\boldsymbol{u}}
\newcommand{\xb}{\boldsymbol{x}}
\newcommand{\yb}{\boldsymbol{y}}
\newcommand{\xib}{\boldsymbol{\xi}}
\newcommand{\zerob}{\boldsymbol{0}}

\newcommand{\Ab}{\boldsymbol{A}}
\newcommand{\Bb}{\boldsymbol{B}}
\newcommand{\Cb}{\boldsymbol{C}}
\newcommand{\Db}{\boldsymbol{D}}
\newcommand{\Ib}{\boldsymbol{I}}
\newcommand{\Kb}{\boldsymbol{K}}
\newcommand{\Qb}{\boldsymbol{Q}}
\newcommand{\Rb}{\boldsymbol{R}}
\newcommand{\Ub}{\boldsymbol{U}}
\newcommand{\Wb}{\boldsymbol{W}}
\newcommand{\Yb}{\boldsymbol{Y}}
\newcommand{\Pib}{\boldsymbol{\Pi}}

\newcommand{\ybs}{\mathbf{y}}

\newcommand{\ubs}{\mathbf{u}}

\newcommand{\Dbc}{\boldsymbol{\mathcal{D}}}
\newcommand{\Obc}{\boldsymbol{\mathcal{O}}}
\newcommand{\Qbc}{\boldsymbol{\mathcal{Q}}}
\newcommand{\Rbc}{\boldsymbol{\mathcal{R}}}
\newcommand{\Tbc}{\boldsymbol{\mathcal{T}}}

\newcommand{\Uc}{\mathcal{U}}
\newcommand{\Xc}{\mathcal{X}}
\newcommand{\Yc}{\mathcal{Y}}
\newcommand{\Zc}{\mathcal{Z}}

\newcommand{\image}{\mathrm{image}}
\newcommand{\rank}{\mathrm{rank}}
\newcommand{\dom}{\mathrm{dom}}

\renewcommand{\boldsymbol}[1]{#1}
\renewcommand{\mathbf}[1]{\mathrm{#1}}

\def\tvdots{\vbox{\baselineskip=2pt \lineskiplimit=0pt \kern6pt \hbox{.}\hbox{.}\hbox{.}}}

\title{\vspace{-2mm}\bf Implicit predictors in regularized data-driven predictive control}
\author{Manuel Kl\"adtke and Moritz Schulze Darup \vspace{2mm}}
\date{}

\fancypagestyle{plain}{%
  \renewcommand{\headrulewidth}{0pt}%
  \fancyhf{}%
  \fancyfoot[C]{\footnotesize 2475-1456 \copyright 2023 IEEE}%
}

\begin{document}
\maketitle

\pagestyle{fancy}
\fancyhf{} 
\renewcommand{\headrulewidth}{0pt}
\fancyfoot[C]{\footnotesize 2475-1456 \copyright 2023 IEEE}

\begin{abstract}
We introduce the notion of implicit predictors, which characterize the input-(state)-output prediction behavior underlying a predictive control scheme, even if it is not explicitly enforced as an equality constraint (as in traditional model or subspace predictive control). To demonstrate this concept, we derive and analyze implicit predictors for some basic data-driven predictive control (DPC) schemes, which offers a new perspective on this popular approach that may form the basis for modified DPC schemes and further theoretical insights.
\end{abstract}
\infoFootnote{M. Kl\"adtke and M. Schulze Darup are with the \href{https://rcs.mb.tu-dortmund.de/}{Control and~Cyberphysical Systems Group}, Faculty of Mechanical Engineering, TU Dortmund University, Germany. E-mails:  \href{mailto:manuel.klaedtke@tu-dortmund.de}{\{manuel.klaedtke, moritz.schulzedarup\}@tu-dortmund.de}. \vspace{0.5mm}}
\infoFootnote{This paper is a \textbf{preprint} of a contribution to the IEEE Control Systems Letters. The DOI of the original paper is \href{https://doi.org/10.1109/LCSYS.2023.3285104}{10.1109/LCSYS.2023.3285104}} 

\section{Introduction}
\label{sec:Introduction}
Data-driven predictive control (DPC) is an increasingly popular control approach that utilizes linear combinations of collected trajectory data to make predictions instead of relying on a system model (see, e.g., \cite{Coulson2019DeePC, Berberich2020, Dorfler2021}). While exact predictions and equivalence to model predictive control (MPC) are typically only established for linear time-invariant (LTI) systems and exact data (with some nonlinear extensions available; see, e.g.,  \cite{Berberich2020, Alsalti2021}), modified versions of DPC show promising results even in the absence of these requirements (e.g., \cite{Coulson2019DeePC, Dorfler2021, Coulson2019RegularizedDeePC, Berberich2020stability}). Most proposed DPC schemes use regularizations, which have demonstrated benefits and many different interpretations (see, e.g., \cite{Coulson2019RegularizedDeePC, Dorfler2021}), including relationships to other approaches such as subspace predictive control (SPC) \cite{FAVOREEL1999}. However, many of these interpretations (e.g. \cite[Thms.~3 and 4]{Breschi2022new}) are based on limit behavior, i.e., regularization weights approaching zero or infinity. In this study, we aim to characterize the behavior of DPC for finite weights, as this is how it is typically applied in practice. To achieve a general characterization, we do not focus on a specific system class or noise properties but, instead, just make (reasonable) assumptions about the data. To make DPC predictions more accessible and relatable to traditional schemes such as MPC, where consistency with input-state-output predictors is explicitly enforced via equality constraints, we introduce the notion of implicit predictors (specified in Def.~\ref{def:implicit_predictor} below). Implicit predictors can be interpreted as the predictive behavior that is implicitly attributed to the data-generating system by the predictive scheme, and we view it as a central and elucidating object with much to learn from. 
To show the benefit of this concept, we derive and analyze implicit predictors for a basic regularized DPC scheme, focusing on two proposed choices of (squared) 2-norm regularization (see \cite{Dorfler2021}), and examine the effect of input and output constraints on these predictors.
The paper is organized as follows. First, in Section~\ref{sec:fundamentals}, we introduce the notion of implicit predictors, specify the objective of this paper, and summarize important preliminaries.
In Section~\ref{sec:implicitPredictorsDPC}, we derive implicit predictors for two types of regularized DPC and analyze the impact of constraints on these predictors. Finally, we conclude our work in Section~\ref{sec:Conclusions} and preview future opportunities we envision for implicit predictors.

\section{Problem statement and preliminaries}\label{sec:fundamentals}

One could argue that (multi-step) predictors, which establish a mapping from an initial state $\xb_0$ and an input (sequence) $\ubs_f$ to a predicted output (sequence) $\ybs_f$, are at the heart of many predictive control schemes. Traditionally, such a predictor $\hat\ybs_f(\xb_0, \ubs_f)$ may be explicitly included in an optimal control problem (OCP) as an equality constraint $\ybs_f=\hat\ybs_f(\xb_0, \ubs_f)$. However, even if not included explicitly (such as in DPC), one may still observe the predictions following a similar pattern, which we formalize as follows.
\begin{defn}\label{def:implicit_predictor}
    We call $\hat\ybs(\xb_0,\ubs_f)$ an \textit{implicit predictor} for an OCP if including the constraint $\ybs_f = \hat\ybs(\xb_0, \ubs_f)$ does not alter the (set of) minimizers $(\ubs_f^\ast, \ybs_f^\ast)$ and the optimal value.
\end{defn}

For DPC, we interpret this definition as the predictive behavior implicitly attributed to the data-generating system by the DPC scheme. Given this interpretation, we take a somewhat opposing viewpoint to the bias-variance hypothesis in \cite[Sec. V.C]{Dorfler2021}. In fact, we claim that the choice of regularizer (and other parameters of the OCP) fully specifies a model class given by the structure of its resulting implicit predictor. Therefore, this concept should be seen as a tool to analyze existing schemes and explain their behavior.
Although this new viewpoint is mainly theoretical in nature, practitioners can use its results to evaluate whether the predictive behavior of a given scheme matches their prior knowledge of the true system properties, and thus select an appropriate scheme similar to traditional model selection. In this spirit, we analyze an existing DPC scheme without making assumptions on the data-generating system class or properties of measurement noise, allowing for a general characterization. 
To prepare this analysis, we recall preliminaries on DPC and its relation to MPC and SPC, while highlighting the role of predictors in these schemes.

\subsection{DPC and its relation to MPC}
Instead of utilizing a discrete-time state-space model 
\begin{subequations}\label{eq:statespace_general}
\begin{align}
    \xb (k+1) &= \fb(\xb(k),\ub(k),k)\\
    \yb(k) &= \gb(\xb(k),\ub(k),k) 
\end{align}
\end{subequations}
with input $\ub\in\R^m$, state $\xb\in\R^n$, and output $\yb\in\R^p$ as in traditional MPC, predictions in DPC are realized based on previously collected trajectory data $(\ubs^{(1)}, \ybs^{(1)}), \hdots, (\ubs^{(\ell)}, \ybs^{(\ell)})$ via linear combinations
$$
    \begin{pmatrix}
        \ubs_\text{pred}\\
        \ybs_\text{pred}
    \end{pmatrix}
    =
    \begin{pmatrix}
        \ubs^{(1)}\\
        \ybs^{(1)}
    \end{pmatrix} a_1 + \hdots + 
    \begin{pmatrix}
        \ubs^{(\ell)}\\
        \ybs^{(\ell)}
    \end{pmatrix} a_\ell
    = 
    \Dbc \ab.
$$
Here, the dimensions of the data matrix $\Dbc\in\R^{L(m+p)\times\ell}$ and generator vector $\ab\in\R^\ell$ are specified by the length~$L$ of recorded (as well as predicted) trajectories and the number~$\ell$ of data trajectories used for predictions. The rationale for this procedure is given by a result for linear time-invariant (LTI) systems, which can be described, e.g., by the specification
\begin{subequations}\label{eq:statespace_LTI}
\begin{align}
    \xb (k+1) &= \Ab\xb(k)+\Bb\ub(k)\\
    \yb(k) &= \Cb\xb(k)+\Db\ub(k), 
\end{align}
\end{subequations}
of \eqref{eq:statespace_general}.
For these LTI systems and assuming $L$ is greater than the lag of the system, $\image(\Dbc)$ is equivalent to the set of all possible system trajectories if and only if \cite{Markovsky2020}
\begin{equation}\label{eq:GPE}
    \rank(\Dbc)=L m + n. 
\end{equation}
 Note that the condition \eqref{eq:GPE} not only signifies a minimum rank for data-driven predictions, but also the maximum rank that the data matrix $\Dbc$ can achieve for exact data. 
For the case that the individual trajectories $(\ubs^{(i)}, \ybs^{(i)})$ are time-shifted sections of a single long trajectory, a popular sufficient condition for \eqref{eq:GPE} is given by Willems' fundamental lemma \cite{WILLEMS2005}. To include the current initial condition of the system as a starting point for predicted trajectories, the predicted I/O-sequence is typically partitioned into a past section $(\ubs_p, \ybs_p)$ and a future section $(\ubs_f, \ybs_f)$ with $N_p$ respectively $N_f$ time-steps yielding
$$
    \begin{pmatrix}
        \ubs_\text{pred}\\
        \ybs_\text{pred}
    \end{pmatrix}
    = 
    \begin{pmatrix}
        \ubs_p\\
        \ubs_f\\
        \ybs_p \\
        \ybs_f
    \end{pmatrix}
    =
    \begin{pmatrix}
        \Ub_p\\
        \Ub_f\\
        \Yb_p \\
        \Yb_f
    \end{pmatrix}\ab
    =
    \Dbc \ab.
$$
The past section of a predicted trajectory is then forced to match the I/O-data $\xib$ recorded in the most recent $N_p$ time-steps during closed-loop operation, i.e., the equality constraints
$$
    \xib= 
    \begin{pmatrix}
        \ubs_p \\ \ybs_p
    \end{pmatrix}
    =
    \begin{pmatrix}
        \Ub_p \\ \Yb_p
    \end{pmatrix}\ab
    = \Wb_p \ab
$$
force any predicted trajectory in \eqref{eq:DPC} to start with the most recently witnessed behavior of the system. In this context, the past trajectory $\xib$ can also be interpreted as the state of a (usually non-minimal) state-space realization of the system and properly specifies its initial condition if $N_p$ is chosen larger or equal to its lag \cite{Markovsky2008}. 
\begin{rem}\label{rem:Statespace}
    Although we have introduced the data-driven predictions in an I/O setting, they can be straightforwardly modified to a state-space setting \cite{DePersis2020}, which we will use for visualization of a low dimensional example in Section~\ref{sec:implicitPredictorsDPC}.
\end{rem}

Now, consider a classical OCP
\begin{align}
\label{eq:MPC}
\min_{\ub(k),\xb(k),\yb(k)} 
 &\sum_{k=0}^{N_f-1} \|\yb(k)\|_{\Qb}^2 +  \|\ub(k)\|_{\Rb}^2 \span \span \\
\nonumber
\text{s.t.} \quad \quad  \xb(0)&=\xb_0, \\
\nonumber
 \xb(k+1)&=\Ab\,\xb(k) + \Bb \ub(k), &&\forall k \in \{0,...,N_f-2\}, \\
 \nonumber
  \yb(k)&=\Cb\,\xb(k) + \Db \ub(k), &&\forall k \in \{0,...,N_f-1\}, \\
 \nonumber
\left(\ub(k),\yb(k)\right) & \in \Uc_k \times \Yc_k, &&\forall k \in \{0,...,N_f-1\}
\end{align}
 for MPC with prediction horizon $N_f$, positive definite weighing matrices $ \Rb\in\R^{m  \times m},\Qb\in\R^{p  \times p }$, convex constraint sets $\Uc_k \subseteq \R^{m}, \Yc_k \subseteq \R^{p}$, as well as equality constraints specifying the initial state condition $\xb_0$ and predicted behavior based on \eqref{eq:statespace_LTI}. It has been shown in \cite{Coulson2019DeePC} that DPC given by
\begin{subequations}
\label{eq:DPC}
\begin{align}
\min_{\ubs_f,\ybs_f,\ab} 
 \|\ybs_f\|_{\Qbc}^2 &+  \|\ubs_f\|_{\Rbc}^2 \label{eq:DPCcost} \\
\text{s.t.} \quad \quad  \begin{pmatrix}
    \xib \\ \ubs_f \\ \ybs_f 
\end{pmatrix} &= \begin{pmatrix}
    \Wb_p \\ \Ub_f \\ \Yb_f 
\end{pmatrix}\ab, \label{eq:DPCeqConstr}\\
\left(\ubs_f, \ybs_f \right) &\in \Uc \times \Yc  \label{eq:DPCsetConstr}
\end{align}
\end{subequations}
with $\Qbc:=\mathrm{blkdiag}(\Qb, ..., \Qb)$, $\Rbc:=\mathrm{blkdiag}(\Rb, ..., \Rb)$, and 
\begin{align*}
    \Yc &:= \left.\left\{\ybs_f
    \in \R^{p N_f}
    \right| \yb(k) \in \Yc_k, \,\forall k \in \{0, ..., N_f-1\}\right\} \\
    \Uc &:= \left.\left\{\ubs_f
    \in \R^{m N_f}
    \right|\ub(k) \in \Uc_k,\,\forall k \in \{0, ..., N_f-1\}\right\}
\end{align*}%
based on exact data generated by an LTI system is equivalent to the MPC in \eqref{eq:MPC}.
Remarkably, and although the original theory behind exact predictions via linear combination of data does not apply to arbitrary nonlinear systems or with noise and disturbances in the data, DPC has shown good closed-loop performance when applied to these cases. A common feature of these successful applications, is the addition of a regularization term $h(\ab)$ to the cost function, where different choices of $h(\ab)$ can have different interpretations for its intended effect on the predictions (\cite{Dorfler2021, Coulson2019DeePC, Coulson2019RegularizedDeePC, Berberich2020stability}).

\subsection{The role of predictors and SPC}

 In the case of linear     MPC, its well-known multi-step predictor can be expressed as
\begin{equation}
    \hat\ybs_{\text{MPC}}(\xb_0,\ubs_f) = \Obc \xb_0 + \Tbc \ubs_f, \label{eq:MPC_predictor}
\end{equation}
where $\Obc$ and $\Tbc$ are often referred to as the extended observability matrix and the impulse response matrix, respectively (see, e.g. \cite{Overschee1996}). Using \eqref{eq:MPC_predictor}, we can state the OCP \eqref{eq:MPC} also as
\begin{align}
\label{eq:MPC_condensed}
\min_{\ubs_f,\ybs_f} 
 \|\ybs_f\|_{\Qbc}^2 +  \|\ubs_f\|_{\Rbc}^2 \span \span \\
\nonumber
\text{s.t.} \quad \quad \ybs_f = \Obc \xb_0 + \Tbc \ubs_f, \quad
\left(\ubs_f, \ybs_f \right) & \in \Uc \times \Yc, \nonumber
\end{align}
which may be useful to eliminate optimization variables and to analyze the structure of its solution. Essentially, instead of using the state-space model \eqref{eq:statespace_LTI} as a one-step predictor for each time-step in \eqref{eq:MPC}, consecutive applications of the one-step predictor yield the multi-step predictor \eqref{eq:MPC_predictor} used in \eqref{eq:MPC_condensed}.  Note that \eqref{eq:MPC_predictor} trivially acts as an implicit predictor for the original MPC problem \eqref{eq:MPC}. By including the constraint $\ybs_f=\hat\ybs_\text{MPC}(\xb_0,\ubs_f)$, other constraints and variables can be eliminated, resulting in \eqref{eq:MPC_condensed}. However, this removal should be seen separately to Definition \ref{def:implicit_predictor} and is not our focus. 

Instead of estimating a state-space model \eqref{eq:statespace_LTI} for MPC, an alternative approach based on subspace identification \cite{Overschee1996} is given by SPC \cite{FAVOREEL1999}. Here, a linear multi-step predictor
\begin{equation}
    \hat\ybs_{\text{SPC}}(\xib,\ubs_f) = \Kb_{\text{SPC}} \begin{pmatrix}
        \xib \\ \ubs_f
    \end{pmatrix}\label{eq:SPC_predictor}
\end{equation}
can be estimated directly from data as the solution 
\begin{equation}
    \Kb_{\text{SPC}} := \argmin_\Kb \left\|\Yb_f-\Kb \begin{pmatrix}
        \Wb_p \\ \Ub_f
    \end{pmatrix} \right\|_F^2 = \Yb_f \begin{pmatrix}
        \Wb_p \\ \Ub_f
    \end{pmatrix}^+ \label{eq:K_SPC}
\end{equation}
to a least squares problem with the Frobenius norm $\|\cdot\|_F$.  
Crucially, this form of SPC is related to DPC as follows.
\begin{lem}
    The SPC predictor \eqref{eq:SPC_predictor} is an implicit predictor for the unregularized DPC problem \eqref{eq:DPC} with exact trajectory data generated by an LTI system. 
\end{lem}
\begin{proof}
    Follows from \cite[Thm.~1]{Fiedler2021}.
\end{proof}
While the result in \cite{Fiedler2021} is actually stated as \eqref{eq:DPC} being equivalent to SPC, we rather rephrase it in the context of implicit predictors. This is because \eqref{eq:DPCeqConstr} cannot be removed (without changing the set of feasible $\xib$) as it also implies $\xib\in\image(\Wb_p)$, which is not captured by the SPC predictor.

 Introducing implicit predictors as explicit equality constraints can be useful to change the solution strategy for some OCPs. This, however, is not our focus in this work.  Instead, we want to highlight this notion as a way to think about what kind of behavior is (implicitly) attributed to the system by DPC schemes, even though the predictor is not given as an equality constraint in any of the following cases. 

\section{Implicit predictors in regularized DPC}\label{sec:implicitPredictorsDPC}

While in the case of deterministic LTI systems, adding additional trajectory data, i.e., columns to $\Dbc$, cannot increase its rank past $\rank(\Dbc)\leq m L+n$, this is typically not the case for systems with noise or nonlinearities. Instead, adding trajectory data generated by an LTI system with some output noise and persistently exciting enough input data (see, e.g., \cite{WILLEMS2005}) will almost surely increase its rank until $\Dbc$ has full row rank for wide enough (at least square) data matrices \cite[Lem.~3]{Breschi2022new} and similar behavior can be observed for nonlinear systems with (or without) noise. Throughout this section, we will thus make the following assumption.
\begin{assum}\label{assum:fullRank}
    The data matrix $\Dbc$ has full row rank.
\end{assum}
 We want to emphasize that Assumption \ref{assum:fullRank} takes into account the presence of measurement noise, and we do not make any further assumptions about the class of system generating the data. 
Once the OCP tuning parameters (including regularization) are determined, the solution becomes deterministic with respect to the provided data $\Dbc$ (including potential output noise, which is included in $\Wb_p$, $\Yb_f$ and cannot be distinguished from the true output data), regardless of the type of system that generated it. Since the aim of this work is to characterize the predictive behavior and not to compare it to the (unknown) true system dynamic (as e.g., in \cite{Breschi2022new}), this allows for a very general analysis. 
 The assumption thus helps in shifting the focus away from what initially caused the lack of rank deficiency and towards what predictions are being made by DPC based on the given data (including noise). 
Furthermore, having more data than necessary in \eqref{eq:GPE} is typically preferred, where some authors report good results with $\Dbc$ being square \cite[Sect.~5.2.4]{Markovsky2021}. As the assumption aligns with this case, we believe it is a reasonable starting point but future work will consider data matrices with ranks between the practical minimum given by \eqref{eq:GPE} and the maximum given by Assumption~\ref{assum:fullRank}.
Crucially, due to Assumption~\ref{assum:fullRank} there is a vector $\ab$ satisfying \eqref{eq:DPCeqConstr} for \textit{any} triple $(\xib, \ubs_f,\ybs_f)$, resulting in arbitrary and meaningless predictions if the scheme is not suitably modified.
However, adding a regularization $h(\ab)$ leading to the regularized DPC scheme
\begin{equation} \label{eq:regularizedDPC}
\min_{\ubs_f,\ybs_f,\ab} 
 \|\ybs_f\|_{\Qbc}^2 +  \|\ubs_f\|_{\Rbc}^2 + h(\ab) \quad 
\text{s.t.}\quad \text{\eqref{eq:DPCeqConstr}--\eqref{eq:DPCsetConstr}}
\end{equation}
has shown good results with respect to its predictive capabilities. Note that, while \eqref{eq:DPCeqConstr} does not restrict the choice of any $(\xib, \ubs_f, \ybs_f)$ in terms of feasibility due to Assumption~\ref{assum:fullRank}, it still defines a relation to $\ab$, which can be used to express the effect of $h(\ab)$ in the following way.
\begin{lem}
    Under Assumption~\ref{assum:fullRank}, the regularized DPC problem \eqref{eq:regularizedDPC} is equivalent to   
    \begin{equation}\label{eq:regDPCouter}
    \min_{\ubs_f,\ybs_f}
     \|\ybs_f\|_{\Qbc}^2 +  \|\ubs_f\|_{\Rbc}^2 + h^\ast(\xib, \ubs_f, \ybs_f) \quad \text{s.t.}\quad  \eqref{eq:DPCsetConstr}
    \end{equation}
    with unique
    \begin{equation}\label{eq:regDPCinner}
    h^\ast(\xib, \ubs_f, \ybs_f) := \min_{\ab}\:\:   h(\ab) \quad 
    \text{s.t.}\quad \eqref{eq:DPCeqConstr}.
    \end{equation}
\end{lem}
\begin{proof}
    The OCP \eqref{eq:regularizedDPC} can be trivially decoupled with the outer problem \eqref{eq:regDPCouter} and inner problem \eqref{eq:regDPCinner} by optimizing over one of the optimization variables while treating the others as parameters. The constraint \eqref{eq:DPCsetConstr} is irrelevant to the inner problem because $(\ubs_f, \ybs_f )$ act as parameters there. On the other hand, the constraint \eqref{eq:DPCeqConstr} is irrelevant to the outer problem because $\ab$ is eliminated by solving the inner problem and \eqref{eq:DPCeqConstr} does not restrict the choice of $(\xib, \ubs_f, \ybs_f)$ due to Assumption~\ref{assum:fullRank}. While the minimizer $\ab^\ast(\xib, \ubs_f, \ybs_f)$ to \eqref{eq:regDPCinner} may be non-unique depending on the choice of $h(\ab)$ and the number $\ell$ of data trajectories (i.e., the width of $\Dbc$), the resulting optimal value of the inner problem $h^\ast(\xib, \ubs_f, \ybs_f)=h(\ab^\ast(\xib, \ubs_f, \ybs_f))$ is trivially unique by optimality.
\end{proof}

While any trajectory $(\xib, \ubs_f, \ybs_f)$ can be predicted via linear combination of collected data trajectories, the regularizer $h(\ab)$ assigns to every such trajectory a (possibly non-unique) optimal generator vector $\ab^\ast(\xib, \ubs_f, \ybs_f)$ and associated unique cost $h^\ast(\xib, \ubs_f, \ybs_f)$. Therefore, the regularizer $h(\ab)$ is, ideally, chosen in a way that incentivizes predictions that seem likely based on the given data in $\Dbc$ by assigning them low (or even zero) cost while avoiding unlikely predictions via high costs.

We will now analyze this behavior for different choices of $h(\ab)$ and characterize the associated implicit predictors. Note that we treat the influence of equality constraints \eqref{eq:DPCsetConstr} on the implied predictor separately in Section~\ref{subsec:ConstrainedDPC}.

\subsection{(Squared) 2-norm regularization}\label{subsec:DPC}

Commonly seen in regularized DPC is (squared) 2-norm regularization, where we choose $h(\ab)=\lambda_a\|\ab\|_2^2$ with weighing parameter $\lambda_a$. Here, the inner problem \eqref{eq:regDPCinner} is a quadratic minimization with linear equality constraints and positive definite Hessian $2 \lambda_a \Ib$, which yields the unique minimizer
\begin{equation}
    \ab^\ast(\xib, \ubs_f, \ybs_f) = \begin{pmatrix}
    \Wb_p \\ \Ub_f \\ \Yb_f 
\end{pmatrix}^+ \begin{pmatrix}
    \xib \\ \ubs_f \\ \ybs_f 
\end{pmatrix} \label{eq:a_min}
\end{equation}
and associated cost
\begin{align*} 
    h^\ast(\xib, \ubs_f, \ybs_f)
    &= \lambda_a \begin{pmatrix}
    \xib \\ \ubs_f \\ \ybs_f 
\end{pmatrix}^\top \!\left(\begin{pmatrix}
    \Wb_p \\ \Ub_f \\ \Yb_f 
\end{pmatrix}\begin{pmatrix}
    \Wb_p \\ \Ub_f \\ \Yb_f 
\end{pmatrix}^\top\right)^{-1}\!\!\begin{pmatrix}
    \xib \\ \ubs_f \\ \ybs_f 
    \end{pmatrix}\!.
    \end{align*}
Utilizing a block LDU decomposition of the weighing matrix, one can further show that 
{\setlength{\jot}{-2pt} 
\begin{align}\label{eq:2-normCost}
     h^\ast(\xib, \ubs_f, \ybs_f) &= \lambda_a (\ybs_f-\hat\ybs_\text{SPC})^\top\Qbc_\text{reg}(\ybs_f-\hat\ybs_\text{SPC}) \\
    &\quad + \lambda_a
    \begin{pmatrix}
        \xib \\
        \ubs_f
    \end{pmatrix}^\top
    \left(
    \begin{pmatrix}
        \Wb_p \\
        \Ub_f
    \end{pmatrix}
    \begin{pmatrix}
        \Wb_p \\
        \Ub_f
    \end{pmatrix}^\top
    \right)^{-1}
    \begin{pmatrix}
        \xib \\
        \ubs_f
    \end{pmatrix}, \nonumber
\end{align}}
where we introduced
\begin{align}
    &\qquad\qquad\;\;\Qbc_\text{reg}:=\left(\Yb_f\left(\Ib-\Pib\right)\Yb_f^\top\right)^{-1},  \nonumber\\
    &\Pib := \begin{pmatrix}
        \Wb_p \\
        \Ub_f
    \end{pmatrix}^\top\left(
    \begin{pmatrix}
        \Wb_p \\
        \Ub_f
    \end{pmatrix}
    \begin{pmatrix}
        \Wb_p \\
        \Ub_f
    \end{pmatrix}^\top
    \right)^{-1}
    \begin{pmatrix}
        \Wb_p \\
        \Ub_f
    \end{pmatrix}, \label{eq:projMatrix}
\end{align}
and occasionally omit the arguments of $\hat\ybs_\text{SPC}(\xib,\ubs_f)$ for brevity. The expression \eqref{eq:2-normCost} highlights the connection between DPC and SPC for this regularization. 
However, this regularizer not only incentivizes predictions that align with the SPC predictor but also introduces a bias associated with the additional quadratic cost term in $(\xib,\ubs_f)$. This (typically unwelcome) bias justifies the introduction of a projection in the regularizer \cite{Dorfler2021}, which we will analyze in Section~\ref{subsec:projectionDPC}. 

Even though predictions that align with the SPC predictor are incentivized with this regularization, the SPC predictor does not act as an implicit predictor in this case. In contrast to the unregularized deterministic case, where the rank deficiency of $\Dbc$ implies a subspace for predictions, the implicit predictor is instead implied by optimality.
\begin{thm}\label{thm:implicitPredictor2Norm}
    Consider the regularized DPC problem \eqref{eq:regularizedDPC} with regularizer $h(\ab)=\lambda_a\|\ab\|_2^2$ and without additional input and output constraints. 
    Under Assumption~\ref{assum:fullRank},
\begin{equation} \label{eq:DPC_predictor2}
        \hat\ybs_\text{DPC}(\xib, \ubs_f)=
        \left(\lambda_a \Qbc_\text{reg}+\Qbc\right)^{-1} \lambda_a \Qbc_\text{reg} \hat\ybs_\text{SPC}(\xib,\ubs_f)
\end{equation}
    is an implicit predictor for this problem.
\end{thm}
\begin{proof}
    The implicit predictor \eqref{eq:DPC_predictor2} is the minimizer 
    $$
        \hat\ybs_\text{DPC}(\xib, \ubs_f)=\argmin_{\ybs_f} 
     \|\ybs_f\|_{\Qbc}^2 +  \|\ubs_f\|_{\Rbc}^2 + h^\ast(\xib, \ubs_f, \ybs_f)
    $$
    to an inner optimization problem, where $\ubs_f$ and $\xib$     act as parameters. Since dropping any terms independent of $\ybs_f$ does not change the minimizer, we can simplify the problem as
    \begin{align*}
        &\argmin_{\ybs_f} && 
     \lambda_a (\ybs_f-\hat\ybs_\text{SPC})^\top\Qbc_\text{reg}(\ybs_f-\hat\ybs_\text{SPC}) 
      + \ybs_f^\top \Qbc \ybs_f \\
     =& \argmin_{\ybs_f} &&  \ybs_f^\top \left(\lambda_a\Qbc_\text{reg}+\Qbc\right)\ybs_f  -2 \hat\ybs_\text{SPC}^\top \lambda_a\Qbc_\text{reg} \ybs_f,
    \end{align*}
    which yields an unconstrained quadratic minimization problem with the minimizer given by \eqref{eq:DPC_predictor2}. Now, since $\hat\ybs_\text{DPC}(\xib, \ubs_f)$ is the parametric minimizer for any $(\xib,\ubs_f)$, the minimizers $(\ubs_f^\ast, \ybs_f^\ast)$ to the regularized DPC problem must naturally satisfy the relation $\ybs_f^\ast = \hat\ybs_\text{DPC}(\xib, \ubs_f^\ast)$ for any $\xib$. Hence, including the equality constraint $\ybs_f = \hat\ybs_\text{DPC}(\xib, \ubs_f)$ with the regularized DPC problem does not change its optimal value or minimizers, making $\hat\ybs_\text{DPC}(\xib, \ubs_f)$ an implicit predictor of this OCP.
\end{proof}
\begin{figure*}
    \centering
        \includegraphics[trim=1.1cm 12.84cm 0.9cm 12.42cm,clip=true, scale=0.9]{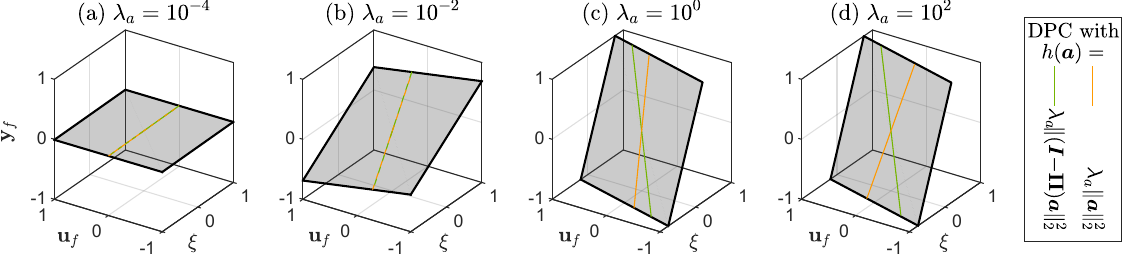}
        \caption{The implicit predictor (grey) is equal for both analyzed unconstrained DPC schemes and its structure is given by a subspace that ``tilts'' between $\hat\ybs_\text{DPC}(\xib, \ubs_f)=0$ and the SPC predictor \eqref{eq:SPC_predictor} depending on $\lambda_a$. The parametric DPC solutions $(\xib, \ubs_f^\ast(\xib), \ybs_f^\ast(\xib))$ for the different regularizations  (green/orange) evolve on this subspace, confirming its validity. }
        \label{fig:unconstrainedDPC}
\end{figure*}

While this implicit predictor does not match the SPC predictor, it still represents a subspace just like the latter and is visualized in Fig.~\ref{fig:unconstrainedDPC} for different regularization weightings~$\lambda_a$. For very high $\lambda_a$, the predictor starts aligning with the SPC predictor (and even matches it for $\lambda_a \to \infty$) as one would expect from the costs given by \eqref{eq:2-normCost}. For lower~$\lambda_a$, the implicit predictor tilts towards $\lim_{\lambda_a \to 0} \hat\ybs_\text{DPC}(\xib, \ubs_f) = \zerob$ as seen in Fig.~\ref{fig:unconstrainedDPC}.a, which is also reasonable, since it represents the optimal solution to $\argmin_{\ybs_f}\|\ybs_f\|_\Qbc^2$. These behaviors perfectly match the viewpoint that ``control and identification regularize each other'' \cite{Dorfler2021} in DPC. However, we emphasize that this means the (implicitly) predicted system behavior of regularized DPC is linear and either consistent with the SPC predictor (for $\lambda_a \to \infty$) or more optimistic than SPC in the sense that it biases these most likely predictions (in the least squares sense, see \eqref{eq:K_SPC}) towards more favorable predictions given by the control objective. While this optimism in the predictions may indeed be one reason why DPC performs better than SPC in some case studies (e.g. \cite[Fig.~2]{Dorfler2021}), it should be recognized that this also leads to potentially undesirable predictive behavior, which we will highlight when discussing additional constraints in Section~\ref{subsec:ConstrainedDPC}.

\subsection{Projection-based (squared) 2-norm regularization}\label{subsec:projectionDPC}

While we only introduced $\Pib$ in \eqref{eq:2-normCost} for notation purposes, we next exploit its role as a projection matrix. However, instead of $\Pib$ itself, we will focus on $\Ib-\Pib$, which is an orthogonal projector on $\ker\big(\begin{pmatrix}
    \Wb_p^\top & \Ub_f^\top
\end{pmatrix}^\top\big)$. Intuitively, these projection matrices can be thought of as splitting any $\ab = \Pib \ab + (\Ib-\Pib )\ab$ into a part $\Pib \ab$ that aligns with the SPC solution and a part $(\Ib-\Pib )\ab$ that does not. In \cite{Dorfler2021}, the regularizer $h(\ab) = \lambda_a \|(\Ib-\Pib)\ab\|_2^2$ was introduced because, although it penalizes deviations from the SPC predictor just as $h(\ab) = \lambda_a \|\ab\|_2^2$, it does not introduce any further bias like the latter in \eqref{eq:2-normCost}. In the following, we will phrase this important result in our framework and subsequently analyze the associated implicit predictor. We first note that the inner problem given by \eqref{eq:regDPCinner} is, again, a quadratic minimization problem with linear constraints and the same (unique) minimizer as before in \eqref{eq:a_min}.
This can easily be verified by plugging $\ab^\ast(\xib, \ubs_f, \ybs_f)$ into the KKT system of \eqref{eq:regDPCinner}. Uniqueness of $\ab^\ast(\xib, \ubs_f, \ybs_f)$ may not be immediately obvious, since the Hessian $2\lambda_a (\Ib-\Pib )^\top (\Ib-\Pib )=2\lambda_a (\Ib-\Pib )$ is only positive semi-definite due to $(\Ib-\Pib)$ being a projector on a subspace of $\R^\ell$. However, since
$$
    \ker\begin{pmatrix}
    \Wb_p \\ \Ub_f \\ \Yb_f 
\end{pmatrix} \subset \ker\begin{pmatrix}
    \Wb_p \\ \Ub_f 
\end{pmatrix}
$$
and $(\Ib-\Pib )$ projects onto the latter kernel, we have $(\Ib-\Pib )\ab \neq \zerob$ for any nonzero $\ab$ in the former kernel, i.e., the Hessian and the constraint matrix only share the trivial nullspace by construction, making the KKT matrix nonsingular and thus $\ab^\ast(\xib, \ubs_f, \ybs_f)$ unique \cite[Sect.~10.1.1]{boyd2004}. The additional cost introduced by the regularizer is therefore given as 
\begin{align*} 
     h^\ast(\xib, \ubs_f, \ybs_f) 
    = \lambda_a \begin{pmatrix}
    \xib \\ \ubs_f \\ \ybs_f 
\end{pmatrix}^\top\!\!\!\!
\begin{pmatrix}
    \Wb_p \\ \Ub_f \\ \Yb_f 
\end{pmatrix}^{+,\top}
\!\!\!\!\!\!\!(\Ib-\Pib)
\begin{pmatrix}
    \Wb_p \\ \Ub_f \\ \Yb_f 
\end{pmatrix}^+\!\!\!\!
\begin{pmatrix}
    \xib \\ \ubs_f \\ \ybs_f 
    \end{pmatrix},
    \end{align*}
which can be shown to be equivalent to
\begin{align}\label{eq:2-normCostProj}
    h^\ast(\xib, \ubs_f, \ybs_f) 
    = \lambda_a (\ybs_f-\hat\ybs_\text{SPC})^\top\Qbc_\text{reg}(\ybs_f-\hat\ybs_\text{SPC}) 
\end{align}
by, again, utilizing block LDU decompositions of the inverses involved in the weighing matrix. The resulting cost is similar to \eqref{eq:2-normCost}, which should not be too surprising, since the projection-based regularization was introduced to remove the bias related to additional costs in $(\xib, \ubs_f)$ and only penalize deviations from $\hat\ybs_\text{SPC}(\xib, \ubs_f)$, which is exactly captured by \eqref{eq:2-normCostProj}. Furthermore, this expression allows us to state the following result regarding the implicit predictor, which is, again, implied by optimality.
\begin{thm}\label{thm:implicitPredictor2NormProj}
    Consider the regularized DPC problem \eqref{eq:regularizedDPC} with regularizer $h(\ab)=\lambda_a\|(\Ib-\Pib)\ab\|_2^2$ and without additional input and output constraints. Under Assumption~\ref{assum:fullRank}, \eqref{eq:DPC_predictor2} is an implicit predictor for this problem.
\end{thm}
\begin{proof}
    The proof can be carried out analogously to Theorem~\ref{thm:implicitPredictor2Norm} by noting that the only difference between the regularizer costs \eqref{eq:2-normCostProj} and \eqref{eq:2-normCost} is irrelevant to the implied predictor since it does not depend on $\ybs_f$.
\end{proof}

This result may be surprising to readers familiar with both DPC regularizations, since the minimizers $(\ubs_f^\ast, \ybs_f^\ast)$ can differ drastically depending on the chosen regularizer. However, it shows that the predictive aspect of both DPC schemes is the same in the sense that, while $(\ubs_f^\ast, \ybs_f^\ast)$ may differ (especially for increasing $\lambda_a$, as seen in Fig.~\ref{fig:unconstrainedDPC}.d), the relation $\ybs_f^\ast = \hat\ybs_\text{DPC}(\xib, \ubs_f^\ast)$ holds for both as visualized in Fig.~\ref{fig:unconstrainedDPC}.
\begin{rem}
    We can easily extend Theorems~\ref{thm:implicitPredictor2Norm} and \ref{thm:implicitPredictor2NormProj} towards reference tracking cost $\|\ybs_f-\ybs_\text{ref}\|_\Qbc^2$, which adds a linear term in the cost function and thus yields an affine implicit predictor, which is biased towards the unconstrained minimum $\ybs_\text{ref}$.
\end{rem}

\subsection{Predictions affected by (output) constraints}\label{subsec:ConstrainedDPC}

As previously noted, the implied predictors associated with regularized DPC tend to predict more optimistically (with respect to the control objective) than the SPC predictor. This phenomenon will be further highlighted in this section, where we investigate the effect of additional constraints.
In the following, we will restrict our attention to a polyhedral set $\Yc$, which is a common assumption in (linear) predictive control. Similar results can be expected for differently shaped sets, although they are not as nice to characterize accurately. Moreover, we introduce $\Xc \subseteq \dom(\xib)$, which should not be seen as a constraint on an optimization variable (because $\xib$ is a parameter of the OCP), but instead can be viewed as the set of realistically occurring initial conditions in practical operation.
\begin{thm}\label{thm:implicitPredictor2NormConstr}
    Consider the regularized DPC problem \eqref{eq:regularizedDPC} with regularizer $h(\ab)=\lambda_a\|\ab\|_2^2$ or $h(\ab)=\lambda_a\|(\Ib-\Pib)\ab\|_2^2$ and with constraint sets $\Uc, \Yc$, where the latter is a polyhedron. Under Assumption~\ref{assum:fullRank},
    an implicit predictor $\hat\ybs_\text{DPC}(\xib, \ubs_f)$ can be characterized as the minimizer to the multiparametric quadratic program~(mpQP)
    \begin{align}\label{eq:PWA_predictor}
    \argmin_{\ybs_f} & \; \ybs_f^\top\! \left(\lambda_a\Qbc_\text{reg}+\Qbc\right)\ybs_f 
      -2 \hat\ybs_\text{SPC}^\top(\xib, \ubs_f) \lambda_a\Qbc_\text{reg} \ybs_f \\
     \qquad \text{s.t.}&\;  \ybs_f \in \Yc \nonumber
    \end{align}
    with parameters $(\xib,\ubs_f)$, which is a continuous piecewise affine (PWA) function on partitions of the domain $\Xc\times \Uc$.
\end{thm}
\begin{proof}
    Derivation of the mpQP can be carried out analogously to Theorems~\ref{thm:implicitPredictor2Norm} and \ref{thm:implicitPredictor2NormProj}, by noting that only the output constraints are explicitly relevant to the implied predictor and other constraints only specify the domain of its parametric solution. The structure of the minimizer follows from standard mpQP results (see, e.g., \cite[Cor.~5.2]{borrelli2017}).
\end{proof}
\begin{figure}
    \centering
    \includegraphics[trim=5.95cm 13.7cm 6.6cm 11.2cm,clip=true, scale=0.9]{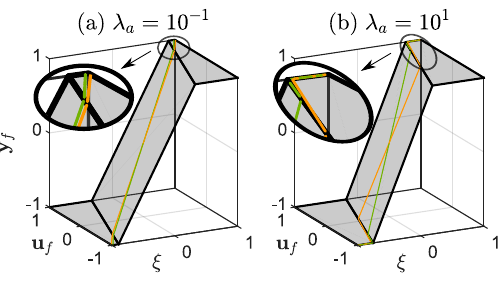}
        \caption{The implicit predictor (grey) for the analyzed DPC schemes with (output) constraints is a PWA function. For increasing $\lambda_a$, the parametric DPC solutions (orange/green, see Fig.~\ref{fig:unconstrainedDPC} for a legend) tend to stay on the segment that matches the unconstrained predictor (if possible).}
        \label{fig:constrainedDPC}
\end{figure}

While PWA functions are not uncommon in (linear) MPC (mainly in relation to its explicit solution \cite{borrelli2017}), we want to stress that the PWA function characterized in Theorem~\ref{thm:implicitPredictor2NormConstr} is in its interpretation very different to an optimal control law $\ubs_f^\ast(\xib)$, since it characterizes the predictive behavior we associate with regularized DPC. While this implicit predictor coincides with \eqref{eq:DPC_predictor2} on the set
    $
        {\Zc := \bigl\{(\xib,\ubs_f)\in\Xc\times \Uc\,\big| \eqref{eq:DPC_predictor2} \; \text{satisfies} \; \hat\ybs_\text{DPC}(\xib,\ubs_f)\in \Yc \bigr\}},
    $
i.e., where the unconstrained solution naturally satisfies the output constraints, it changes drastically wherever this coincidence does not occur, as seen in Fig.~\ref{fig:constrainedDPC}. However, at this point, a discussion on the uniqueness of implicit predictors is due. While the implicit predictors derived so far are kept as general as possible (i.e., they are independent of input cost weights $\Rbc$ and constraints $\Uc$), by Definition~\ref{def:implicit_predictor} agreement with the predictor needs to occur only at optimal points $(\ubs_f^\ast, \ybs_f^\ast)$ for it to be valid. In other words, if the actual solution to the constrained regularized DPC problem \eqref{eq:regularizedDPC} naturally stays on $\Zc$, i.e., we have $(\xib, \ubs_f^\ast)\in \Zc$ for all $\xib\in\Xc$, then its predictive behavior is equally well explained by both the linear predictor \eqref{eq:DPC_predictor2} and the PWA solution to \eqref{eq:PWA_predictor}. Moreover, the DPC solution favors this phenomenon for increasing $\lambda_a$, as highlighted in Fig.~\ref{fig:constrainedDPC}, since deviations of the unconstrained predictor, in this case, are increasingly penalized by the regularizer. However, it should be obvious that (even with very large $\lambda_a$) a match of both predictors cannot occur for such $\xib$, where \eqref{eq:DPC_predictor2} cannot be satisfied for any $(\ubs_f,\ybs_f)\in \Uc \times \Yc$. In that case, the DPC solution will definitely change its predictive behavior in favor of satisfying the output constraints in accordance with \eqref{eq:PWA_predictor}. In other words, regularized DPC as in \eqref{eq:regularizedDPC} will simply provide a ``very unlikely'' solution if all ``more likely'' solutions (in the sense of proximity to its unconstrained predictions or even the SPC predictor) are infeasible. Consequentially, from the viewpoint of an explicit solution to the OCP \cite{borrelli2017}, where $\Xc$ acts as a constraint of the parameter $\xib$, its set of feasible states is equal to $\Xc$ regardless of other constraints $\Uc, \Yc$.

Finally, we want to highlight that the implicit predictor in Theorem~\ref{thm:implicitPredictor2NormConstr} is independent of the input constraints, since $\Uc$ only affects its domain but, in contrast to $\Yc$, not its structure. 
Therefore, the linear implicit predictor \eqref{eq:DPC_predictor2} remains valid in the absence of output constraints (i.e., $\Yc = \R^{p N_f}$), even if additional inputs constraints $\Uc \subset \R^{m N_f}$ are considered.

\subsection{A remark on the numerical example}

To visualize the implicit predictors, we deliberately chose a very low dimensional example with $n = p = m = N_f = \Qbc=\Rbc = 1, \ell = 3$ and treated it in a state-space setting (see Rem.~\ref{rem:Statespace}).
Furthermore, the constraint sets involved in Fig.~\ref{fig:constrainedDPC} are given by $\Uc=\Yc=[-1,1]$. The data generating system is LTI with parameters $(\Ab, \Bb, \Cb, \Db) = (2, -0.5, 1, 0)$ and zero-mean Gaussian measurement noise with variance $\sigma^2 = 0.01$. However, we would like to emphasize that the structure of the implicit predictor does not depend on the data generating system class but only the data itself and the OCP parameters.

\section{Conclusions and Outlook}
\label{sec:Conclusions}

By introducing the notion of implicit predictors, we related the input-(state)-output prediction behavior of regularized DPC schemes to more traditional predictive control schemes such as MPC and SPC, where the predictor is explicitly enforced as an equality constraint. The structure of these implicit predictors seems very relevant to us since they can be interpreted as the behavior that the data-driven predictions attribute to the system based on given data and independently of its actual (unknown) behavior. To demonstrate this concept, we derived implicit predictors for a basic regularized DPC scheme and analyzed their structure for two popular choices of (squared) 2-norm regularization and its dependence on constraints.

In future work, we will continue this structural analysis with general ranks beyond Assumption~\ref{assum:fullRank} and further common DPC modifications such as slack variables \cite{Coulson2019DeePC, Berberich2020stability}, (terminal) equality constraints \cite{Berberich2020stability}, and different regularization choices such as the 1-norm \cite{Coulson2019DeePC,Dorfler2021} or general $p$-norms. Although currently limited to providing a novel perspective on current DPC practices, we believe that these analyses have the potential to establish a foundation for new schemes that promote advantageous properties of the implicit predictor, and for new theoretical insights, including DPC stability proofs such as in \cite{Berberich2020stability}, but based on the implicit predictor.


\end{document}